\documentclass[10pt]{article}
\usepackage{balance}
\usepackage{amsmath,amssymb,amsfonts}
\usepackage{amsthm}
\usepackage{subfigure}
\usepackage{enumitem}
\usepackage{authblk}
\usepackage{graphicx}

\begin{document}





\title{EZ-AG: Structure-free data aggregation in MANETs using push-assisted self-repelling random walks}

\author[1]{V. Kulathumani\thanks{vinod.kulathumani@mail.wvu.edu}}
\author[1]{M. Nakagawa\thanks{manakagawa@mix.wvu.edu}}
\author[2]{A. Arora\thanks{anish.arora@samraksh.com}}
\affil[1]{Department of Computer Science, West Virginia University, Morgantown WV 26505}
\affil[2]{The Samraksh Company, Dublin OH 43017}

\newtheorem{theorem}{Theorem}[section]
\newtheorem{lemma}[theorem]{Lemma}
\newtheorem{proposition}[theorem]{Proposition}
\newtheorem{Corollary}[theorem]{Corollary}
\newtheorem{Definition}[theorem]{Definition}

\maketitle

\begin{abstract}
This paper describes EZ-AG, a structure-free protocol for duplicate insensitive data aggregation in MANETs. The key idea in EZ-AG is to introduce a token that performs a self-repelling random walk in the network and aggregates information from nodes when they are visited for the first time. A self-repelling random walk of a token on a graph is one in which at each step, the token moves to a neighbor that has been visited least often. While self-repelling random walks visit all nodes in the network much faster than plain random walks, they tend to slow down when most of the nodes are already visited. In this paper, we show that a {\em single step push phase at each node} can significantly speed up the aggregation and eliminate this slow down. By doing so, EZ-AG achieves aggregation in only $O(N)$ time and messages. In terms of overhead, EZ-AG outperforms existing structure-free data aggregation by a factor of at least $log(N)$ and achieves the lower bound for aggregation message overhead. We demonstrate the scalability and robustness of EZ-AG using ns-3 simulations in networks ranging from $100$ to $4000$ nodes under different mobility models and node speeds. We also describe a hierarchical extension for EZ-AG that can produce multi-resolution aggregates at each node using only $O(N log N)$ messages, which is a poly-logarithmic factor improvement over existing techniques.
\end{abstract}



\section{Introduction}
The focus of this paper is on computing order and duplicate insensitive data aggregates (also referred to as ODI-synopsis) and delivering them to every node in a mobile ad-hoc network (MANET) \cite{cascade_ipsn07, fisheye_gossip, snapshots, synopsis_diffusion}. In an ODI synopsis, the same data can be aggregated multiple times but the result is unaffected. MAX and MIN are natural examples of such duplicate insensitive data aggregation. Other statistical aggregates such as COUNT and AVERAGE can also be implemented with ODI synopsis using probabilistic techniques \cite{synopsis_diffusion, sensor-databases}. We are specifically motivated by data aggregation requirements in extremely large scale mobile sensor networks such as networks of UAVs, military networks and dense vehicular networks, where the number of nodes are often several thousands. 


In static sensor networks and networks with stable links, data aggregation can be performed by routing along fixed structures such as trees or network backbones \cite{naik_sprinkler, tag, ctp, directdiff}. However, in MANETs, routing has proven to be quite challenging beyond scales of a few hundred nodes primarily because topology driven structures are unstable and are likely to incur a high communication overhead for maintenance in the presence of node mobility\cite{rtswjournal}. Therefore, structure-free techniques are more appropriate for data aggregation in MANETs. However, a simple technique like flooding data from each node to every other node is not scalable as it incurs an overall cost of $O(N^2)$, where $N$ is the number of nodes in the network. Therefore, in this paper we explore the use of self-repelling random walks as a structure free method for data aggregation. 

Random walks are appropriate for data aggregation in mobile networks because they are inherently unaffected by node mobility. The idea is to introduce a token in the network that successively visits all nodes in the network using a random walk traversal and computes the overall aggregate. We say that a node is {\em visited} by a token when the node gets exclusive access to the token; the visitation period can be used by the node to add node-specific information into the token, resulting in data aggregation.  Note that the concept of visiting all nodes individually differs from that of token dissemination \cite{mmgossip, trickle} over the entire network where it suffices for every node to simply hear at least one token, as opposed to getting exclusive access to a token. 

Note, however that traditional random walks may be too slow in visiting all nodes in the network because they may get stuck in regions of already visited nodes. Hence, in this paper we consider self-repelling random walks \cite{self-repelling}. A self-repelling random walk is one in which at each step the walk moves towards one of the neighbors that has been least visited \cite{self-repelling} (with ties broken randomly). Self-repelling random walks were introduced in the $1980s$ and have been studied extensively in the physics literature. One of the striking properties of self-repelling random walks is the remarkable uniformity with which they visit nodes in a graph. This has been studied formally in terms of the variance of the number of visits at each node during a self-repelling random walk. The study in \cite{rw12} shows that this variance is bounded by values less than $1$ even in lattices of dimensions $2048 \times 2048$. The uniformity property of self-repelling random walks is interesting for data aggregation because it implies that the token is likely to spread towards unvisited areas in the network as opposed to getting stuck in already visited regions. 

Indeed, our results in this paper confirm that until about $85\%$ coverage, duplicate visits are very rare with self-repelling random walks highlighting the efficiency with which a majority of nodes in the network can be visited without extra overhead. However, we observe a slow down when going towards $100\%$ coverage. When most of the nodes are already visited, the token executing self-repelling random walk has to explore the graph to find the next unvisited node. As a result, at $100\%$ coverage, we observe that the exploration overhead (ratio of number of token transfers to the number of unique node visited) rises to around $2$ at a network size of $4000$ nodes. Moreover, we observe a small rise in exploration overhead as a function of network size. This trend is logarithmic indicating that the total number of steps to finish visiting all nodes grows as $O(Nlog(N))$ where $N$ is the number of nodes in the network. To correct this shortcoming, we introduce a complementary {\em push} phase that speeds up the convergence of the random walk. The {\em push} phase consists of just one message from each node: before the random walk is started, each node announces its own state to all its neighbors. Thus, after the push phase, each node now carries information about all its neighbors. As a result, when the random walk executes, it does not have to visit all nodes to finish the aggregation. In fact, we show that the aggregation can finish {\em before} the slow down starts for the self-repelling random walk. As a result both the aggregation time and number of messages are now bounded by $O(N)$.

\subsection{Summary of contributions}
\begin{itemize}[leftmargin=*]

\item We introduce a novel structure-free technique for data aggregation in MANETs that exploits properties of self-repelling random walks and complements it with a push phase. We find that {\em a little push goes a long way} in speeding up aggregation and reducing message overhead. In fact, the push phase consists of just a single message from each node to its neighbors. By adding this push phase, we show that both the aggregation time and number of messages are bounded in EZ-AG by $O(N)$. In fact, we show that aggregation is completed in significantly less than $N$ token transfers. The protocol is thus extremely simple, requires very little state maintenance (each nodes only remembers the number of times it has been visited), requires no network structures or clustering. 

\item We compare our results with structure-free techniques for ODI data aggregation such as gossiping and show a $log(N)$ factor improvement in messages compared to existing gossip based techniques. We evaluate our protocol using simulations in ns-3 on networks ranging from $100$ to $4000$ nodes under various mobility models and node speeds. We also evaluate and compare our protocol with a prototype tree-based technique for data aggregation (i.e., structure based) and show that our protocol is better suited for MANETs and remains scalable under high mobility. In fact, {\em the performance of EZ-AG improves as node mobility increases.} 

\item Finally, we also provide an extension to EZ-AG which supplies multi-resolution aggregates to each node. In networks that are quite large, providing each node with only a single aggregate may not be sufficient. Hierarchical EZ-AG addresses this issue by providing each node with multiple aggregates of neighborhoods of increasing size around itself using only $O(N log N)$ messages. Moreover, we also show that aggregates of nearby regions can be obtained at a progressively faster rate than farther regions. Hierarchical EZ-AG outperforms existing techniques for multi-resolution data aggregation by a factor of $log^{4.4}N$.

\end{itemize}

\subsection{Outline of the paper}

In Section~\ref{sec:related}, we describe related work and specifically compare our contributions with existing work in structure-based protocols, structure-free protocols and random walks. In Section~\ref{sec:model}, we state the system model. In Section~\ref{sec:protocol}, we describe the EZ-AG protocol.  In Section~\ref{sec:analysis}, we analytically characterize the bound on messages and time for EZ-AG. In Section~\ref{sec:hierarchy}, we describe a hierarchical extension for EZ-AG. In Section~\ref{sec:eval}, we describe the results of our evaluation using ns-3 and compare EZ-AG with a  prototype tree-based protocol for data aggregation. We conclude in Section~\ref{sec:conclusion}.

\section{Related work}
\label{sec:related}

\subsection{Structure-based protocols}
The problem of data aggregation and one-shot querying has been well studied in the context of static sensor networks. It has been shown that in-network aggregation techniques using spanning trees and network backbones are efficient and reliable solutions for the problem \cite{naik_sprinkler, tag, ctp, directdiff}. However, in the context of a mobile network, such fixed routing structures are likely to be unstable and could potentially incur a high communication overhead for maintenance \cite{rtswjournal}.  In this paper, we have systematically compared EZ-AG with a prototype tree-based technique for data aggregation and have shown that it outperforms the tree-based idea in mobile networks. We notice that the improvement gets progressively more significant as the average node speed increases.


\subsection{Structure-free protocols}
Flooding, neighborhood gossip and spatial gossip are three structure-free techniques that can be used for data aggregation. Note that flooding data from all nodes to every other node has a messaging cost of $O(N^2)$. Alternatively, one could use multiple rounds of neighborhood gossip where in each round a node averages the current state of all its neighbors and this procedure is repeated until convergence \cite{gossip-survey, randgossip}. However, this method requires several iterations and has also been shown to have a communication cost and completion time of $O(N^2)$ for convergence in grids or random geometric graphs, where connectivity is based on locality \cite{rabbatgossip}. 

In \cite{fisheye_gossip} and \cite{cascade_ipsn07}, a spatial gossip technique is described where each node chooses another node in the network (not just neighbors) at random and gossips its state. When this is repeated $O(log^{1+\epsilon}N)$ times, all nodes in the network learn about the aggregate state. Note that this scheme requires $O(N.polylog(N))$ messages. Our random walk based protocol, EZ-AG, requires only $O(N)$ messages. Note also that while all this prior work is on static networks, we demonstrate our results on MANETs.

\subsection{Random walks}
Random walks and their cover times (time taken to visit all nodes) have been studied extensively for different types of {\em static} graphs \cite{rw1, rw7}. In this paper, we are specifically interested in time varying graphs that are relevant in the context of mobile networks.

Self-avoiding and self-repelling random walks are variants of random walks which bias the walk towards unvisited nodes \cite{self-repelling}. The unformity in coverage of such random walks in 2-d lattices has been pointed out in \cite{rw12}. Our paper extends the analysis of self-repelling random walks presented in \cite{rw12} for application in MANETs that are modeled as time varying random geometric graphs. Further, we show that by complementing self-repelling random walks with a push phase, we can complete aggergation in $O(N)$ time and messages. The idea of locally biasing random walks and its impact in speeding up coverage has been pointed out in \cite{rw8} for static networks. Self-repelling random walks are different than the local bias technique presented in \cite{rw8}. Moreover, we show how to improve the convergence of self-repelling random walks using a complementary push-phase and demonstrate our results on mobile networks. 

In a recent paper \cite{census-arxiv}, we have addressed the problem of duplicate-sensitive aggregation using self-repelling random walks and in that solution we have used a gradient technique to speed up self-repelling random walks. The short temporary gradients introduced in \cite{census-arxiv} are used to {\em pull} the token towards unvisited nodes so that each node is visited at least once. The solution in \cite{census-arxiv} requires $O(Nlog(N))$ messages. In this paper, we address duplicate insensitive aggregation and show that it can be achieved using self-repelling random walks with just $O(N)$ messages.

\section{Model}
\label{sec:model}

\subsection{Network model}

We consider a mobile network of $N$ nodes modeled as a geometric Markovian evolving graph \cite{clementi1}. Each node has a communication range $R$. We assume that the $N$ nodes are independently and  uniformly deployed over a square region of sides $\sqrt{A}$ resulting in a network density $\rho = N/A$ of the deployed nodes. Consider the region to be divided into square cells of sides $R/\sqrt{2}$. Thus the diagonal of each such cell is the communication range $R$. Let $R^2 > 2clog(N)/\rho$. It has been shown that there exists a constant $c>1$ such that each such cell has $\theta(log N)$ nodes whp, i.e., the degree of each node is $\theta(log N)$ whp. Such graphs are referred to as geo-dense geometric graphs \cite{rw8}. Denote $d = \theta(log N)$ as the degree of connectivity.

The objective of the protocol is to compute a duplicate insensitive aggregate of the state of nodes in a MANET. The aggregate could be initiated by any of the nodes in the MANET or by a special static node such as a base station that is connected to the rest of the nodes. The aggregate needs to be disseminated to all nodes in the network. The protocol could be invoked in a one-shot or periodic aggregation mode. 

\subsection{Mobility model}

We consider $3$ different mobility models for our evaluations. 

\begin{itemize}[leftmargin=*]

\item The first is a random direction mobility model (with reflection) \cite{rw10, rw11} for the nodes. This is a special case of the random walk mobility model \cite{mobilitymodels-survey}. In this mobility model, at each interval a node picks a random direction uniformly in the range $[0,2\pi]$ and moves with a constant speed that is randomly chosen in the range $[v_l,v_h]$. At the end of each interval, a new direction and speed are calculated. If the node hits a boundary, the direction is reversed. Motion of the nodes is independent of each other. An important characteristic of this mobility model is that it preserves the uniformity of node distribution: given that at time $t = 0$ the position and orientation of users are independent and uniform, they remain uniformly distributed for all times $t > 0$ provided the users move independently of each other \cite{rw11, clementi1}. 

\item The second is random waypoint mobility model. Here, each mobile node randomly selects one location in the simulation area and then travels towards this destination with constant velocity chosen randomly from $[v_l,v_h]$ \cite{mobilitymodels-survey}. Upon reaching the destination, the node stops for a duration defined by the {\em pause time}. After this duration, it again chooses another random destination and the process is repeated. We set the pause time to $2$ seconds between successive changes. 

\item The third is Gauss Markov mobility model. In this model, the velocity of mobile
node is assumed to be correlated over time and modeled as a Gauss-Markov stochastic process \cite{mobilitymodels-survey}. We set the temporal dependence parameter $\alpha=0.75$. Velocity and direction are changed every $1$ second in the Gauss Markov Model. 

\end{itemize}

We consider node speeds in the range of $3$ to $21$ m/s. For the deployment density that we have chosen, a mapping between node speed and the average link changes per node per second is listed in Table~\ref{tab1}. This table quantifies the link instability caused by node mobility at different node speeds. As seen in Table~\ref{tab1}, because of high network density, the network structure is rapidly changing at the speeds chosen for evaluation. 

\begin{table}[ht]
\caption{Mapping between speed and link changes per node per second} 
\centering 
\begin{tabular}{c c c c c} 
\hline\hline 
Size & 3m/s & 9m/s & 15m/s & 21m/s \\ [0.5ex] 
\hline 
100 & 1 & 5 & 7 & 9 \\ 
200 & 2 & 6 & 9 & 12 \\
300 & 2 & 7 & 10 & 14 \\
500 & 3 & 8 & 12 & 16 \\
1000 & 3 & 9 & 14 & 18 \\ 
2000 & 3.8 & 10 & 16 & 20 \\
4000 & 4 & 12 & 18 & 23 \\ [1ex] 
\hline 
\end{tabular}
\label{tab1} 
\end{table}


\subsection{Metrics}

A key metric that we are interested in is the number of times the token is transferred to already visited nodes. We present this in the form of {\em exploration overhead} which is defined as the ratio of the number of token transfers to the number of unique nodes whose data has been aggregated into the token. We compute exploration overhead at different stages of coverage as the random walk progresses.

Typically, random walks are evaluated in terms of their {\em cover times}, which is defined as the time required to visit all nodes. For a standard random walk, the notion of physical time, messages and the number of steps are all equivalent. However, for the push assisted self-repelling random walks these are somewhat different. The total number of messages required to complete the data aggregation includes the push messages, the messages involved in the self-repelling random walk and the messages involved in disseminating the result to all the nodes using a flood. Moreover, each token transfer step itself consists of announcement, token request and token transfer messages. Thus, although proportional, the number of messages is different than the number of token transfer steps. Hence we separately characterize the number of messages during empirical evaluation. 

Finally we note that since we study random walks on mobile networks, the notion of time is also related to node speed. Moreover, when dealing with wireless networks, time also involves messaging delays. Therefore, during empirical evaluation we separately characterize the actual convergence time (in seconds) along with the number of steps (i.e., number of token transfers).

\section{Protocol}
\label{sec:protocol}

The node requesting the aggregate first initiates a flood in the network to notify all nodes about the interest in the aggregate. Note that each node broadcasts this flood message exactly once.  This results in $N$ messages. 

Once a node receives this request, it {\em pushes} its state to its neighbors. Each node uses the data received from its neighbors to compute an aggregate of the state of all its neighbors. Note that the push phase also requires exactly $N$ messages.

Soon after the initiator sends out an aggregate request, it also initiates a token to perform a  self-repelling random walk. A node that has the token broadcasts an {\em announce} message. Nodes that receive the announce message reply back with a token {\em request} message and include the number of times they have been visited by the token in this request. The node that holds the token selects the requesting node which has been visited least number of times (with ties broken randomly) and transfers the token to that node. This token transfer is repeated successively. Note that nodes which hear a token announcement schedule a token request at a random time $t_r$ within a bounded interval, where $t_r$ is proportional to the number of times that they have been visited. Thus nodes that have not been visited or visited fewer times send a request message earlier. When a node hears a request from a node that has been visited fewer or same number of times, it suppresses its request. Thus, the number of requests received for a token announcement remains fairly constant and irrespective of network density.

In the following section, we prove analytically that the aggregate can be computed from all nodes in the network whp in $O(N)$ token transfers. In the empirical evaluation, we show that the median number of token transfers is actually only $kn$, where $0 < k < 1$, and $k$ is unaffected by network size. Thus, the median exploration overhead is less than $1$. One can use this observation to terminate the self-repelling random walk after exactly $N$ steps and whp one can expect that data from all the nodes has been aggregated.


Once the aggregate has been computed, the result can simply be flooded back to all the nodes. This requires $O(N)$ messages. Another potential solution (when aggregate is only required at a base station) is to transmit the aggregated tokens using a long distance transmission link (such as cellular or satellite links) in hybrid MANETs where the {\em long links} are used for infrequent, high priority data.  

The protocol is thus extremely simple, requires very little state maintenance (each node only remembers the number of times it has been visited), requires no network structures or clustering.

\section{Analysis}
\label{sec:analysis}

In this section, we first show that the aggregation time and message overhead for push assisted self-repelling random walks is $O(N)$. We consider a static network for our analysis. In section~\ref{sec:eval}, we evaluate the protocol under different mobility models and verify that the results hold even in the presence of mobility.

First, we state the following claim regarding the uniformity in the distribution of visited nodes during the progression of a self-repelling random walk.

\begin{proposition}
The distribution of visited nodes (and unvisited nodes) remains spatially uniform during the progression of a self-repelling random walk.
\end{proposition}

\noindent {\em Argument:}~~Our claim is based on the analysis of uniformity in coverage of self-repelling random walks in \cite{rw12} and in \cite{srrw-short-arxiv}. In \cite{rw12}, the variance in the number of visits per node of {\em self-repelling} random walks is shown to be tightly bounded, resulting in a uniform distribution of visited nodes across the network. More precisely, let $n_i(t,x)$ be the number of times a node $i$ has been visited, starting from a node $x$. The quantity studied in \cite{rw12} is the variance $(1/N)(\sum_i (n_i(t,x) - \mu)^2)$, where $\mu = (1/N)(\sum_i  n_i(t,x))$. It is seen that this variance is bounded by values less than $1$ even in lattices of dimensions $2048 \times 2048$. A detailed extension of this analysis for {\em mobile networks} is presented in Section~\ref{sec:uniformity} which shows the uniformity with which nodes are visited during a self-repelling random walk. We use this to infer that even after the walk started, the distribution of visited nodes (and by that token, unvisited nodes) remains uniform. The result shows that the self-repelling random walk is not stuck in regions of already visited nodes - instead, it spreads towards unvisited areas. \rlap{$\qquad \Box$} 

\begin{theorem}
The required number of messages for data aggregation by EZ-AG in a connected, static network of $N$ nodes with uniform distribution of node locations is $O(N)$. \label{thm1}
\end{theorem}

\begin{proof}
We note that the aggregation request flood and the result dissemination flood require $O(N)$ messages. During the push phase, each node broadcasts its state once and this also requires only $N$ messages. Now, we analyze the self-repelling random walk phase.

\begin{figure}[t]
  \begin{center}
    \includegraphics[width=.5\textwidth]{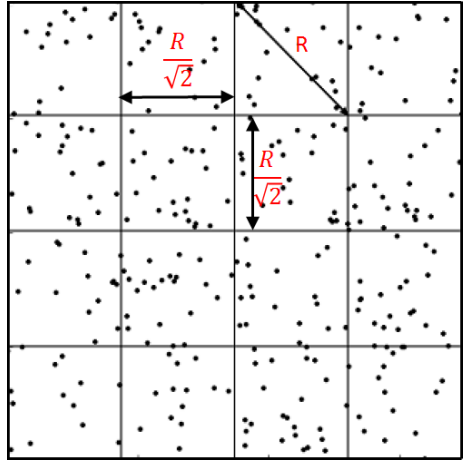}
    \caption{Proof synopsis: Consider the region divided into square cells with diagonal size $R$. At the end of single step push phase, each node has information about all nodes in its cell. So it is sufficient for the token (performing a self-repelling random walk) to visit one node in each cell to finish aggregation. }
    \label{fig:1}
  \end{center}
  \vspace{-3mm}
\end{figure}

Consider the region to be divided into square cells of sides $R/\sqrt{2}$ (see Fig~\ref{fig:1}). Thus the diagonal of each such cell is the communication range $R$. Recall from our system model that each such cell has $\theta(log N)$ nodes whp at all times and there are $O(N/log(N))$ such cells. Therefore, at the end of the push phase, each node has aggregated information about its $\theta(log N)$ cell neighbors. Also note that the network can be divided into $\theta(N/log(N))$ sets of nodes that each contain information about $\theta(log(N))$ nodes within their cell. Therefore, the self-repelling random walk has to visit at least one node in each cell to finish aggregating information from all nodes. 

To analyze the number of token transfers required to visit at least one node in each cell, we use the analogous coupon collector problem (also known as the double dixie cup problem) which studies the expected number of coupons to be drawn from $B$ categories so that at least $1$ coupon is drawn from each category \cite{coupon1}. To ensure that at least $1$ coupon is drawn from each category whp, the required number of draws is $O(B log B)$. Using this result and the fact that a self-repelling random walk traverses a network uniformly, we infer that $O((N/log N)*log(N/logN))$ token transfers are needed to visit at least $1$ node in each of the $\theta(N/log N)$ cells. 

Note that $log(N) > log(N/log(N))$. Hence, the required number of messages for the push assisted self-repelling random walk based aggregation protocol is $O(N/log(N)*log(N))$, i.e., $O(N)$.
\end{proof}

Note that in the presence of mobility, the node locations with respect to cells may not be preserved during the push phase. Therefore the generation of $\theta(N/log(N))$ identical partitions of network state as described in the above analysis may not exactly hold. However, in section~\ref{sec:eval} we empirically ascertain that $kN$ token transfers (where $k < 1$) are still sufficient to aggregate data from all nodes even in the presence of mobility. In fact, we observe that {\em the required token transfers actually decrease with increasing speed}, indicating that data aggregation using self-repelling random walks is {\em actually helped by mobility}.

It follows from the above result that the total time for aggregation is also $O(N)$. The impact of network effects such as collisions on the message overhead and aggregation time (if any) will be evaluated in Section~\ref{sec:eval}.

\section{Extension for hierarchical aggregation}
\label{sec:hierarchy}

When a network is quite large, providing each node with only a single aggregate for the entire network may not be sufficient. It may be more useful to provide each node with a distance-sensitive multi-resolution aggregate, i.e., to provide aggregate information about neighborhoods of increasing sizes around itself. In this section, we describe how EZ-AG can be extended to provide such multi-resolution synopsis of nodes in a network with only $O(N log N)$ messages. 

Existing techniques for such hierarchical aggregation require $O(N log^{5.4}N)$ messages \cite{cascade_ipsn07}. Thus, EZ-AG offers a poly-logarithmic factor improvement in terms of number of messages for hierarchical aggregation. Moreover, EZ-AG can also be used to generate hierarchical aggregates that are distance-sensitive in refresh rate, where aggregates of nearby regions are supplied at a faster rate than farther neighborhoods. 

\subsection{Description}

We divide the network into square cells at different levels ($0$, $1$, .. $P$) of exponentially increasing sizes (shown in Fig.~\ref{fig:hierarchy}). At the lowest level (level $0$), each cell is of sides $R / \sqrt{2}$. Recall from our system model that each such cell has $\theta(log(N))$ nodes whp. For simplicity, let us denote $\theta(log(N)$ by the symbol $\delta$. Thus, there are $N / \delta$ cells at level $0$. Note that $4$ adjoining cells of level $i$ constitute a cell of level $i+1$. Thus, each cell at level $j$ has $\delta 4^j$ nodes whp. At the highest level $P$, there is only one cell with all the $N$ nodes. Note that $P = log_4(N/\delta)$. At any given time, a node belongs to one cell at each level.

\begin{figure}[t]
  \begin{center}
    \includegraphics[width=.5\textwidth]{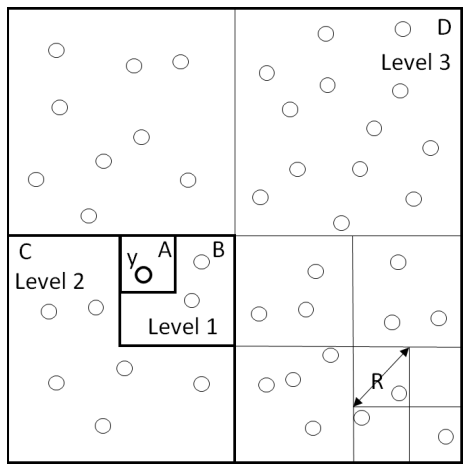}
    \caption{Extension of EZ-Ag to deliver multi-resolution aggregates: The network is partitioned into cells of increasing hierarchy where the cell at smallest level is of diagonal $R$. The node $y$ shown in the figure would receive an aggregate corresponding to one cell at each level that it belongs to. In this case, it would receive aggregates for cells A, B, C and D. The largest cell $D$ consists of the entire network. }
    \label{fig:hierarchy}
  \end{center}
\end{figure}

To deliver multi-resolution aggregates, we introduce a token and execute EZ-AG at each cell at every level. A token for a given cell is only transferred to nodes within that cell and floods its aggregate to nodes within that cell. Thus, there are $N/\delta$ instances of EZ-AG at level $0$ and each instance computes aggregates for $\delta$ nodes, i.e., $\theta(log N)$ nodes.  

The computation and dissemination of aggregates by different instances of EZ-AG are not synchronized. Thus, a node may receive aggregates of different levels at different times. Also, since the nodes are mobile, an aggregate at level $l$ received by a node at any given time corresponds to the cell of the same level $l$ in which it resides at that instant.

\begin{figure*}[htbp]

  \begin{center}
    \mbox{
      \subfigure[Distribution of number of visits at $50\%$ coverage] {\scalebox{0.4}{\includegraphics[width=\textwidth]{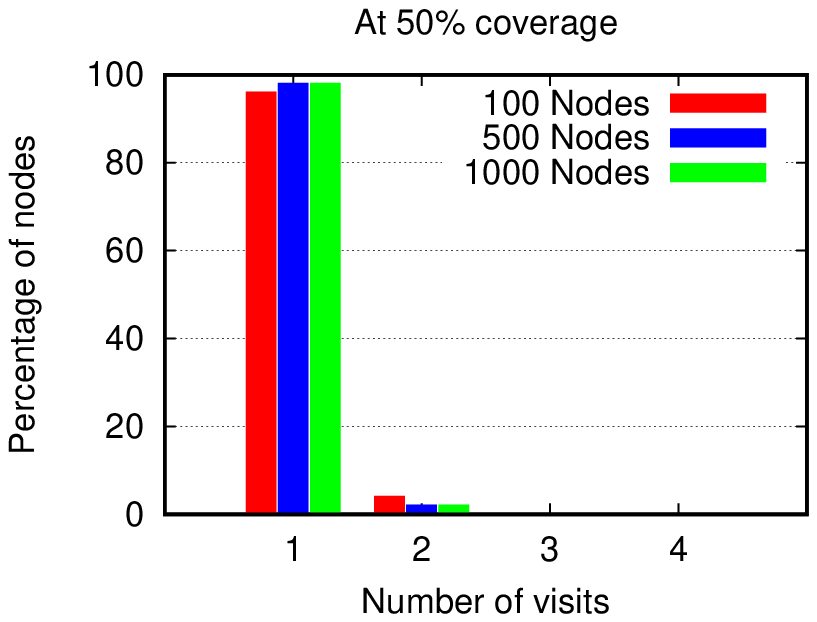}} \label{fig:1a}} \quad
      \subfigure[Distribution of number of visits at $75\%$ coverage] {\scalebox{0.4}{\includegraphics[width=\textwidth]{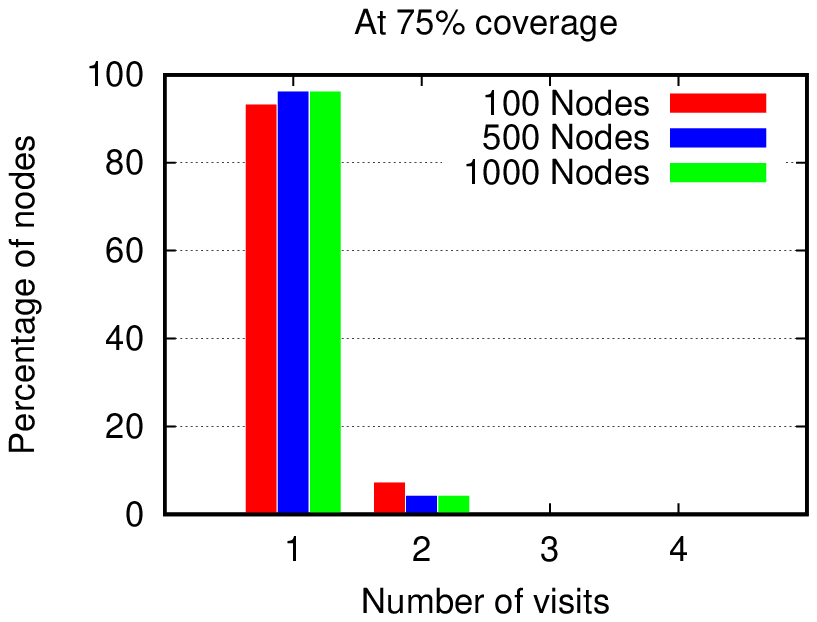}} \label{fig:1b}} 
      } \\
    \mbox{
      \subfigure[Distribution of number of visits at $85\%$ coverage] {\scalebox{0.4}{\includegraphics[width=\textwidth]{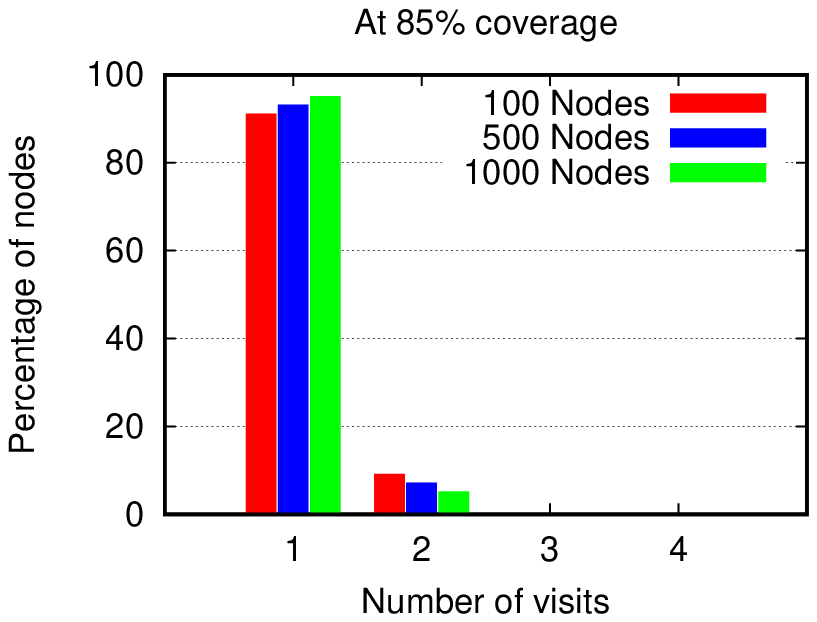}} \label{fig:1c}} \quad
      \subfigure[Distribution of number of visits at $100\%$ coverage] {\scalebox{0.4}{\includegraphics[width=\textwidth]{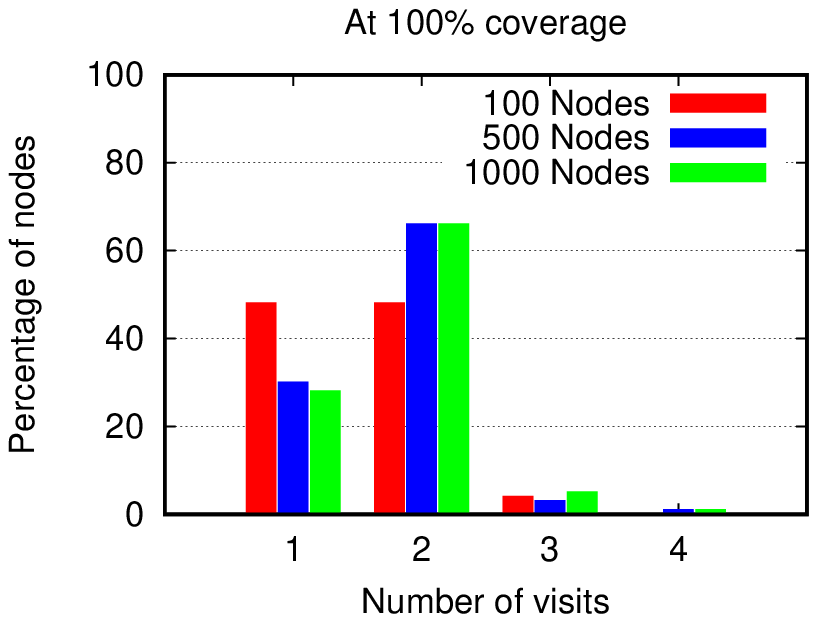}} \label{fig:1d}}     
      } 
    \caption{Distribution of number of visits at each node at different stages of exploration of a self-repelling random walk (network size 100,400 and 1000 nodes)}
       \label{fig:1-1}
  \end{center}
\vspace*{-4mm}
\end{figure*}

\begin{figure*}[htbp]
  \begin{center}
    \mbox{
      \subfigure[Distribution of  number of visits to each node in a pure random walk] {\scalebox{0.4}{\includegraphics[width=\textwidth]{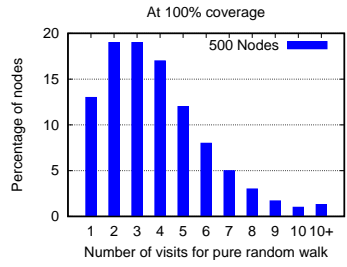}} \label{fig:4a1}} \quad
      \subfigure[Distribution of  number of visits to each node in a self-repelling random walk ] {\scalebox{0.4}{\includegraphics[width=\textwidth]{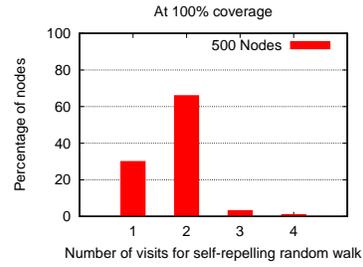}} \label{fig:4b1}} 
          
      } 
    \caption{{Comparison of coverage uniformity with pure random walks}}
       \label{fig:4-1}
  \end{center}
\end{figure*}

\begin{figure*}[htbp]
  \begin{center}
    \mbox{
      \subfigure[Exploration overhead as a function of percentage of nodes visits for self-repelling random walks] {\scalebox{0.4}{\includegraphics[width=\textwidth]{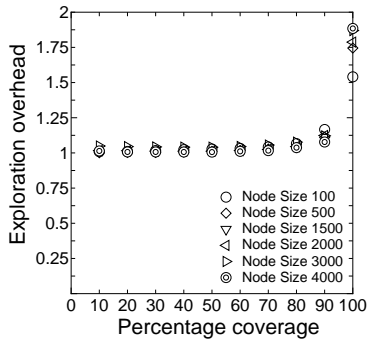}} \label{fig:2a}} \quad
      \subfigure[Exploration overhead at $100\%$ coverage as a function of network size] {\scalebox{0.4}{\includegraphics[width=\textwidth]{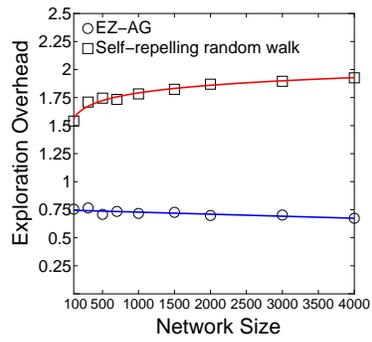}} \label{fig:2b}}         
      } 
    \caption{{Analysis of exploration overhead}}
       \label{fig:2}
  \end{center}
\end{figure*}

\subsection{Analysis}

\begin{theorem}
An ODI aggregate at level $j$ can be computed using hierarchical EZ-AG in $O(4^j \delta)$ time and messages.
\end{theorem}

\begin{proof}
Note that each cell at level $j$ contains $\theta(4^j \delta)$ nodes whp. Therefore, using Theorem~\ref{thm1}, EZ-AG only requires $O(4^j \delta)$ time and messages to compute aggregate within the cell.
\end{proof}

We note from the above theorem that aggregates at level $0$ can be published every $O(\delta)$ time, aggregates at level $1$ can be published every $O(4 \delta)$ time and so on. Thus, aggregates for cells at smaller levels can be published exponentially faster than those for larger cells. Thus, if the tokens repeatedly compute an aggregate and disseminate within their respective cells, EZ-AG can generate hierarchical aggregates that are distance-sensitive in refresh rate, where aggregates of nearby regions are supplied at a faster rate than farther neighborhoods.

\begin{theorem}
Hierarchical EZ-AG can compute an ODI aggregate for all cells at all levels using $O(N log N)$ messages.
\end{theorem}

\begin{proof}
Note that a cell at level $0$ contains $\delta$ nodes and there are $N / \delta$ such cells. The aggregate for cells at level $0$ can be computed using $O(\delta)$ messages. 

In general, there are $N / 4^j \delta$ cells at level $j$ and aggregates for these cells can be computed using $O(4^j \delta)$ messages. Summing up from levels $0$ to $P$, the total aggregation message cost ($M$) for hierarchical EZ-AG can be computed as follows.

\begin{eqnarray*}
M &=& \sum_{j=0}^{P} 4^j \delta \frac{N}{4^j \delta} \\
  &=& \sum_{j=0}^{P} N \\
  &=& O(N log N) \qed
\end{eqnarray*}  
\end{proof}

Thus, hierarchical EZ-AG can compute an ODI aggregate for all cells at all levels using $O(N log N)$ messages.

\section{Performance evaluation}
\label{sec:eval}

In this section, we systematically evaluate the performance of EZ-AG using simulations in ns-3. We set up MANETs ranging from 100 to 4000 nodes using the network model described in Section~\ref{sec:model}. Nodes are deployed uniformly in the network with a  deployment area and communication range such that $R^2 = 4log(N)/\rho$. Thus, the network is geo-dense with $c=2$, i.e., each node has on average $2log(N)$ neighbors whp and the network is connected whp. We test such networks in our simulations with the following mobility models: 2-d random walk, random waypoint and Gauss-Markov (described in Section~\ref{sec:model}). The average node speeds range from $3$ to $21$ m/s. We also consider static networks as a special case.

First, we analyze the convergence characteristics of the push-assisted self-repelling random walk phase in EZ-AG and compare that with self-repelling random walks and plain random walks. Next, we analyze the total messages and time taken by EZ-AG. Finally, we compare EZ-AG with a prototype tree based protocol and with gossip based techniques.

\subsection{Coverage uniformity}
\label{sec:uniformity}

First, in Fig.~\ref{fig:1a}, Fig.~\ref{fig:1b} and Fig.~\ref{fig:1c},we show the number of times each node is visited when the self-repelling random walk has finished visiting $50\%$ of the nodes, $75\%$ of the nodes and $85\%$ of the nodes. We observe that most of the nodes are just visited once and this result holds even at $1000$ nodes. These graphs {\em highlight the uniformity with which nodes are visited as self-repelling random walks progress}. The self-repelling random walk is not stuck in regions of already visited nodes - instead, it spreads towards unvisited areas. Otherwise, one would have observed more duplicate visits to the previously visited nodes.In Fig.~\ref{fig:1d}, we analyze the distribution of number of visits at each node when $100\%$ coverage is attained. Here, we see that most nodes are visited $2$ or $3$ times and the distribution falls off rapidly after that. 

We then compare the uniformity in coverage with that of pure random walks. In Fig.~\ref{fig:4a1}, we plot the number of visits to each node until all nodes are visited at least once for a $500$ node network. In comparison with Fig.~\ref{fig:4b1}, we observe that the tail of the distribution is much longer and the number of duplicate visits is much higher for pure random walks.

\subsection{Convergence characteristics}
\label{sec:conv}

\begin{figure*}[htbp]
  \begin{center}
    \mbox{
      \subfigure[Impact of mobility models on exploration overhead of EZ-AG] {\scalebox{0.3}{\includegraphics[width=\textwidth]{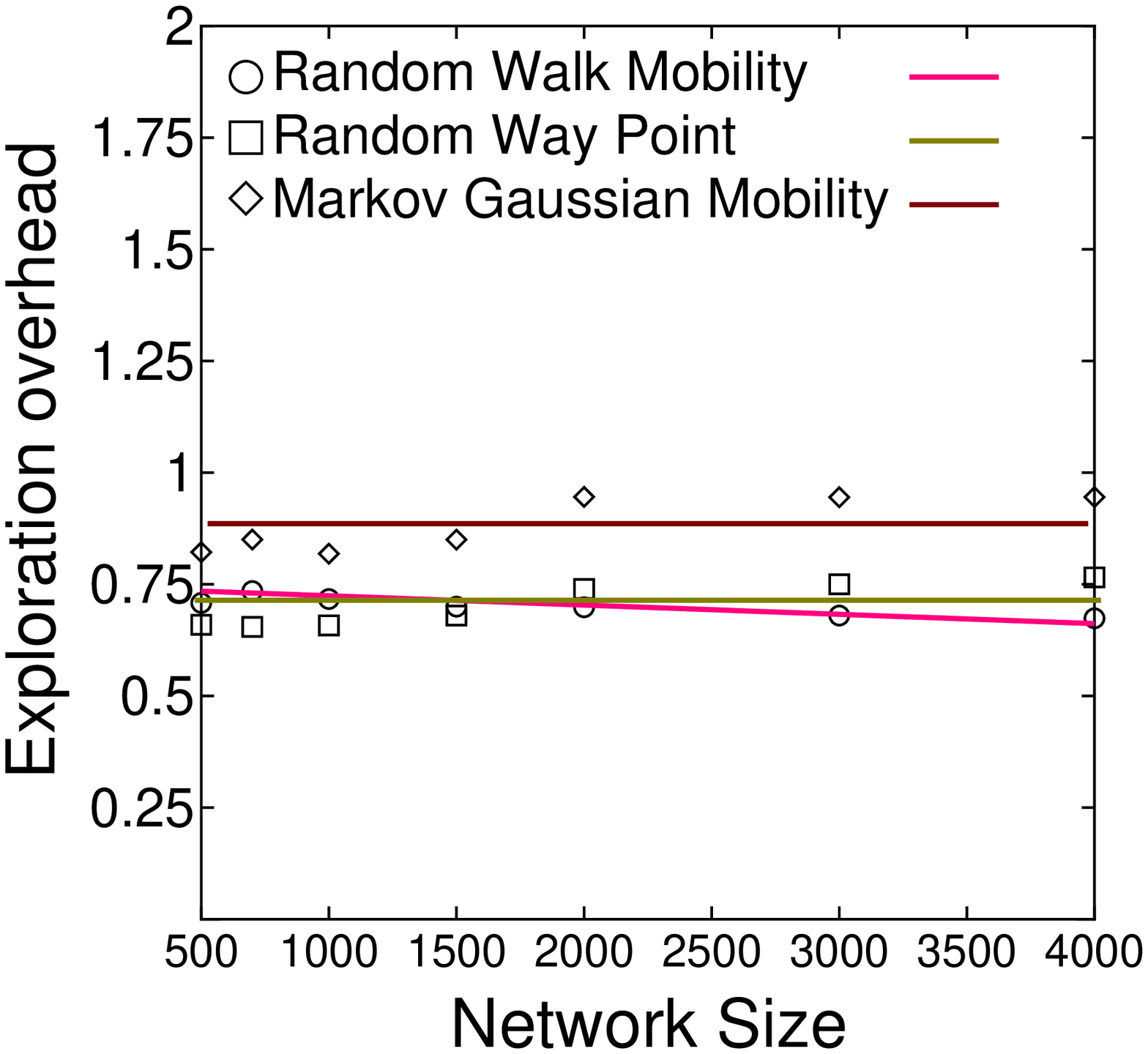}} \label{fig:3a}} \quad
      \subfigure[Impact of node speed on exploration overhead of EZ-AG)] {\scalebox{0.3}{\includegraphics[width=\textwidth]{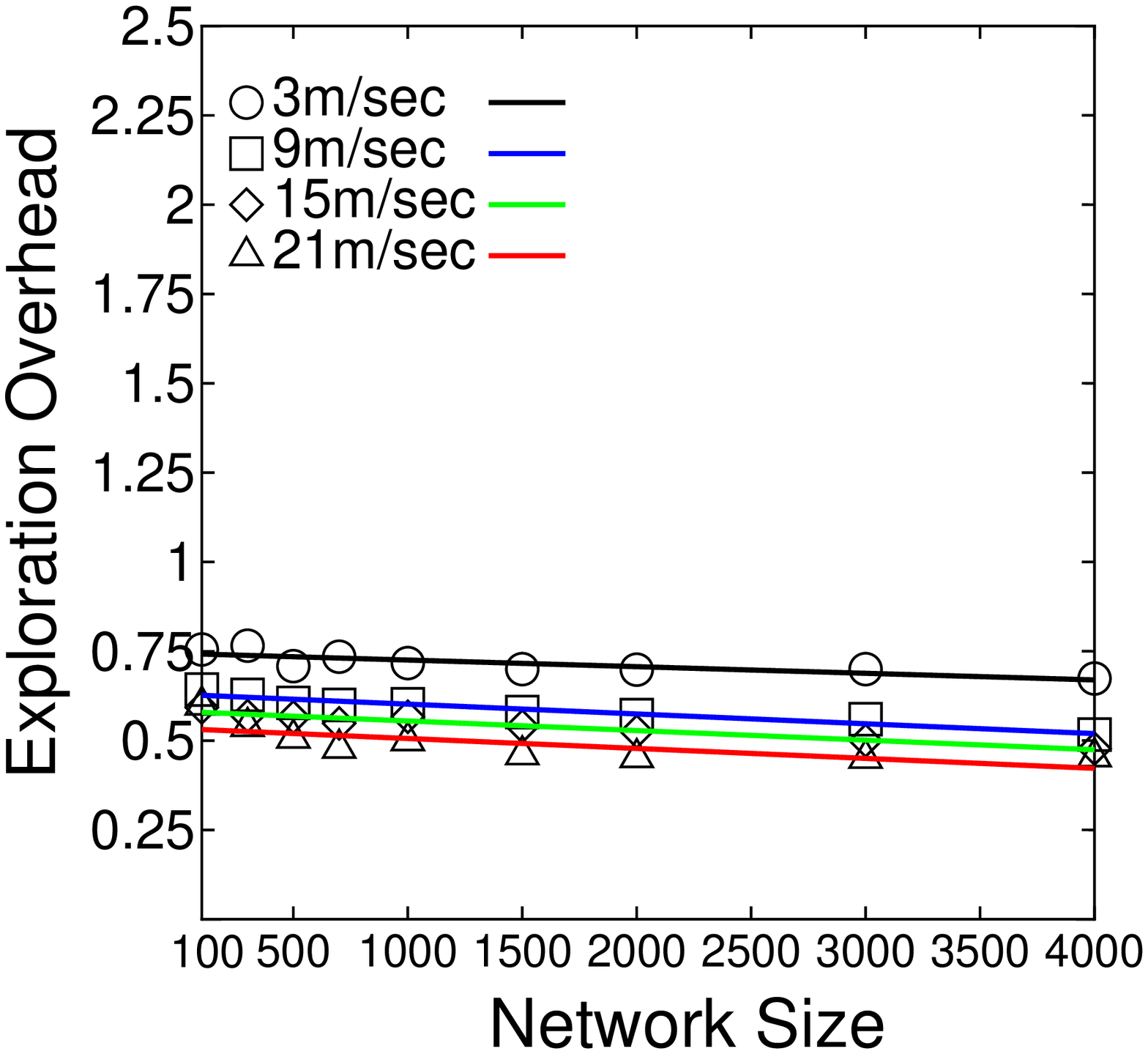}} \label{fig:3b}} 
      \subfigure[Exploration overhead as a function of node speed ] {\scalebox{0.3}{\includegraphics[width=\textwidth]{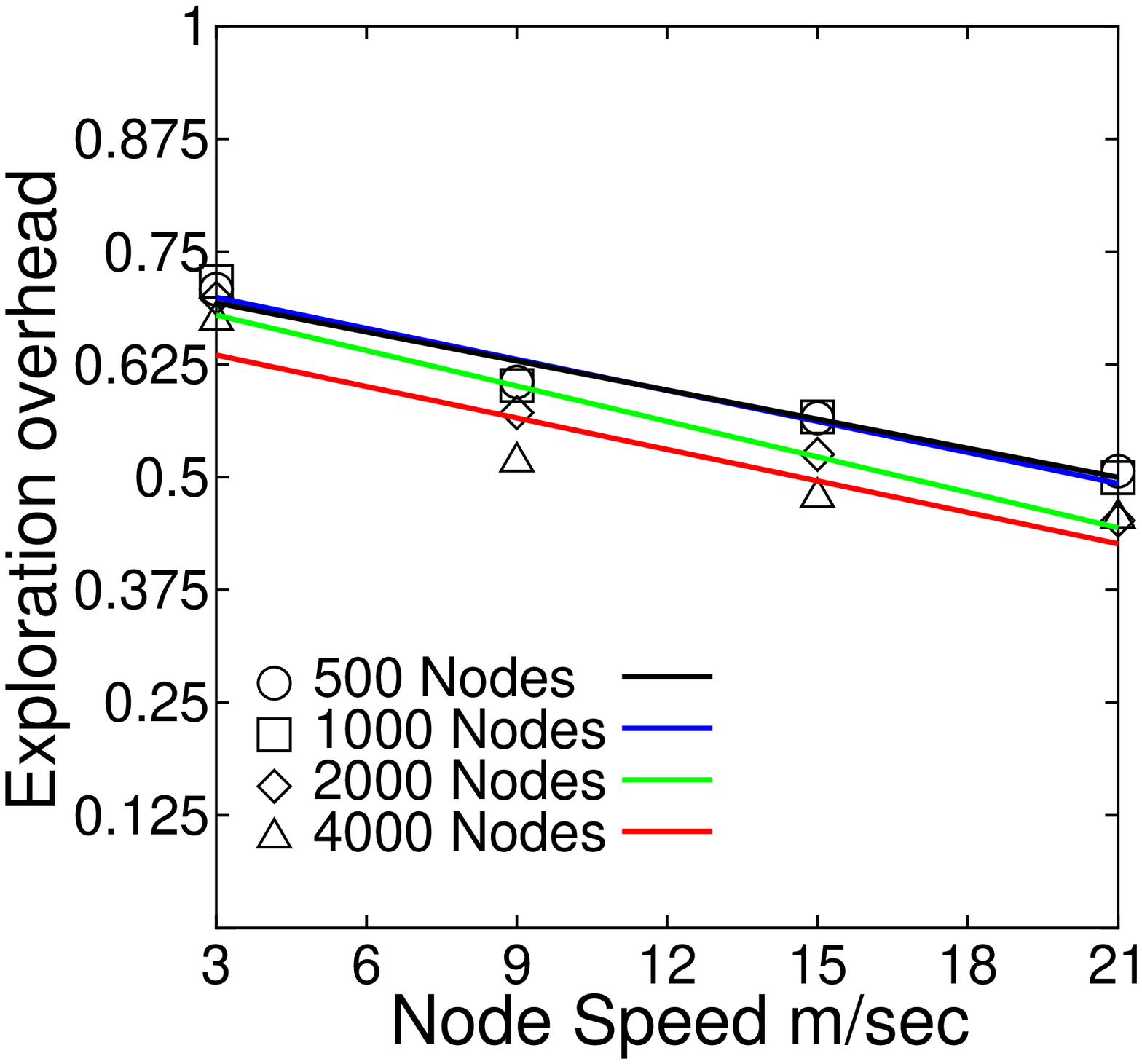}} \label{fig:4}}         
      } 
    \vspace{-1mm}  
    \caption{{Analysis of exploration overhead of EZ-AG for different mobility models and network speeds}}
       \label{fig:3}
  \end{center}
\vspace*{-4mm}
\end{figure*}

Next, in Fig.~\ref{fig:2a}, we show the exploration overhead of self-repelling random walk during different stages of coverage. As seen in Fig.~\ref{fig:2a}, until about $85\%$ coverage, self-repelling random walks have an exploration overhead of around $1$ (irrespective of network size) but then the overhead starts to rise sharply. This is because, until this point self-repelling enables a token to find an unvisited node directly and there are very few wasted explorations. A slowdown for self-repelling random walk is noticed after this point. As a result, the exploration overhead at $100\%$ coverage is close to $2$ and moreover it increases with network size. This is what we aim to address using EZ-AG.


The exploration overhead at $100\%$ coverage  is shown in Fig.~\ref{fig:2b} for self-repelling random walks and EZ-AG (i.e., push-assisted self-repelling random walks). As seen in the figure, the exploration overhead for self-repelling random walks grows with a logarithmic trend due to the wasted explorations towards the tail end of the random walk phase when most of the nodes are already visited. The push assisted self-repelling random walks remove these wasted explorations and as a result the median exploration overhead stays constant at all network sizes and is actually less than $1$ (approximately $0.75$ as seen in Fig.~\ref{fig:2b}). 

\subsection{Variance and terminating condition}
\label{sec:term}

In Fig.~\ref{fig:var}, we show the variation in exploration overhead for EZ-AG over $50$ different trials at different network sizes. We observe that irrespective of network size, for $97.5\%$ of the trials, the exploration overhead is smaller than $1$. We can use this to design a terminating condition for the random walk phase of the protocol. For example, we could terminate the random walk phase after exactly $N$ steps, and then start the dissemination of the aggregate.

\begin{figure}[t]
  \begin{center}
    \includegraphics[width=.55\textwidth]{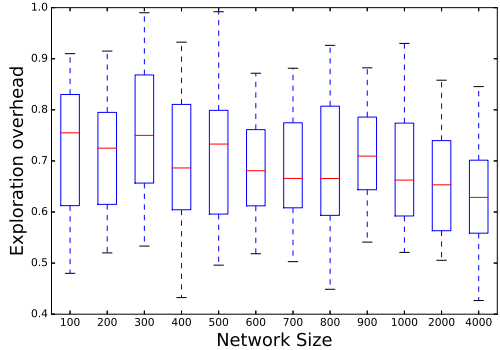}
    \caption{Variance in exploration overhead}
    \label{fig:var}
  \end{center}
  \vspace{-3mm}
\end{figure}

\subsection{Impact of mobility and speed}
\label{sec:speed}

In Fig.~\ref{fig:3a} and Fig.~\ref{fig:3b}, we evaluate the impact of mobility model and network speed on the exploration overhead of push assisted self-repelling random walks. We observe that even though random waypoint and Gauss Markov models do not preserve the uniform distribution of node locations, the exploration overhead exhibits a similar trend. As seen in Table~\ref{tab1}, the network structure is rapidly changing at the speeds chosen for evaluation. Despite this, in Fig.~\ref{fig:3b}, we observe that the exploration overhead actually starts decreasing with node speed (this is shown more clearly in Fig.~\ref{fig:4} for networks with different sizes).

\begin{figure}[ht]
  \begin{center}
    \includegraphics[width=.55\textwidth]{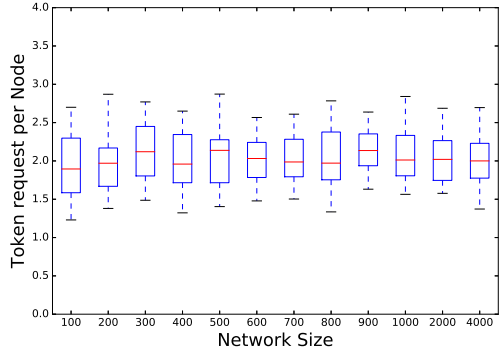}
    \caption{Number of token requests generated per token transfer}
    \label{fig:5c}
  \end{center}
  \vspace{-4mm}
\end{figure}

\subsection{Messages and Time}
\label{sec:msgtime}

\begin{figure*}[htbp]
  \begin{center}
    \mbox{
      \subfigure[Total messages as a function of network size] {\scalebox{0.4}{\includegraphics[width=\textwidth]{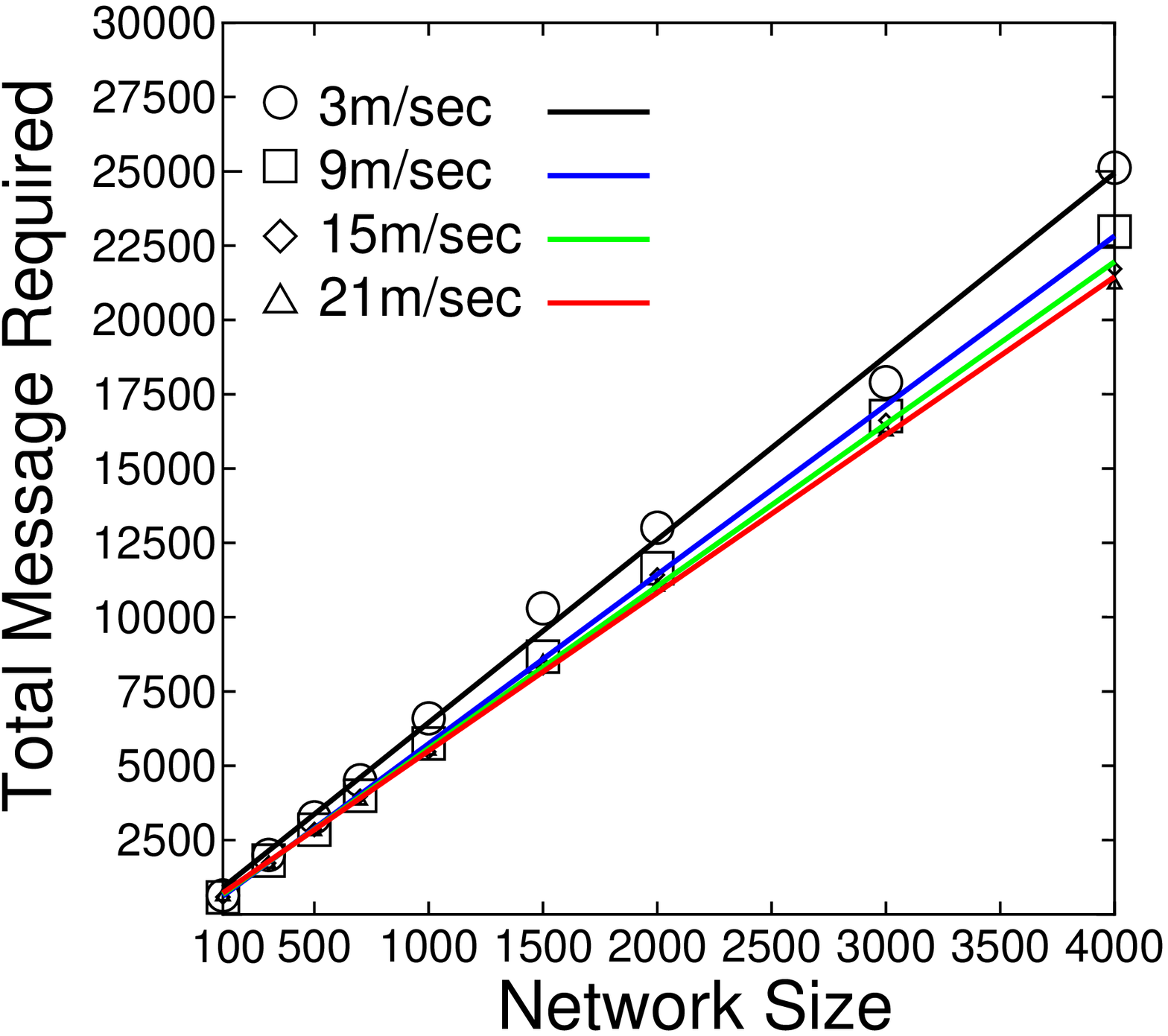}} \label{fig:5a}} \quad
      \subfigure[Aggregation time as a function of network size] {\scalebox{0.39}{\includegraphics[width=\textwidth]{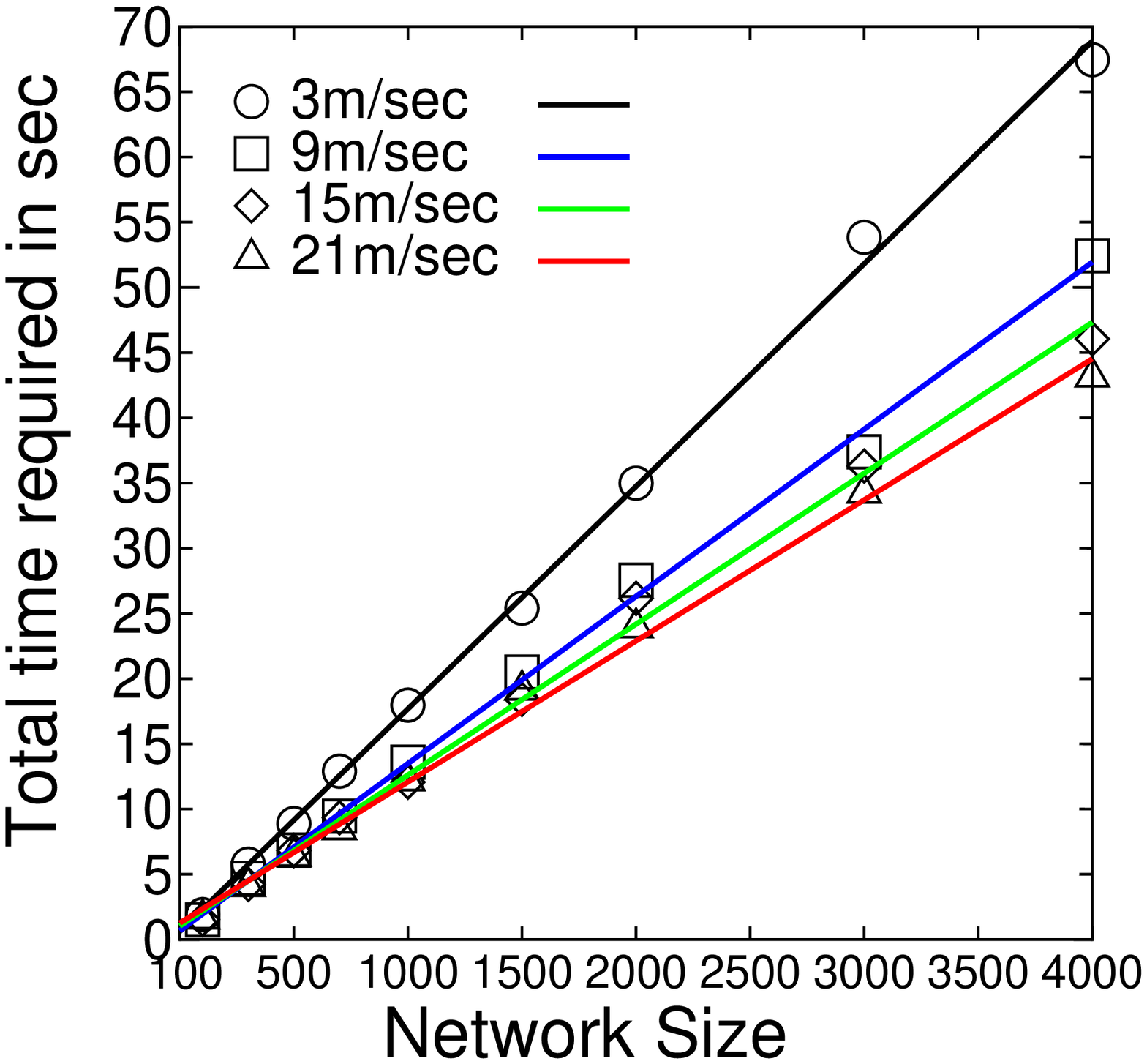}} \label{fig:5b}}         
      } 
    \vspace{-1mm}  
    \caption{{Analysis of time and messages for EZ-AG}}
       \label{fig:5}
  \end{center}
\vspace*{-4mm}
\end{figure*}

In Fig.~\ref{fig:5a} and Fig.~\ref{fig:5b}, we show the total number of messages and the total aggregation time as a function of network size for the aggregation protocol based on push-assisted self-repelling random walks. The total number of messages required to complete the data aggregation includes the push messages, the messages involved in the self-repelling random walk phase and the messages involved in disseminating the result to all the nodes using a flood. Note that, each token transfer step itself consists of announcement, token request and token transfer messages. These are all included in Fig.~\ref{fig:5a} which shows that the messages grow linearly with network size. 

An interesting aspect of the token transfer procedure is the number of requests generated for a token during each iteration. Note that the average number of neighbors increases as $\theta(log N)$ when the network size increases.  However, from Fig.~\ref{fig:5c}, the number of token requests per transfer is seen to be independent of the number of neighbors. From the box plot of Fig.~\ref{fig:5c}, we observe that the average number of token requests in each trial is in the range of $1-3$. This is because nodes that are visited less often send a request earlier than those that are visited more times. And, if a node hears a request from a node that has been visited less often than itself, it suppresses its request. Thus, irrespective of the neighborhood density, the number of token requests per node stay constant.

As seen in Fig.~\ref{fig:5b}, the total aggregation time also exhibits a linear trend. Note that the measurement of time is quite implementation specific and incorporates messaging latency in the wireless network. For instance, in our implementation each transaction (i.e., each iteration of token announcement, token requests and token passing) took on average $25ms$. But this number could be much smaller using methods such as \cite{collabcounting} that use collaborative communication for estimating neighborhood sizes that satisfy given predicates.

\subsection{Comparison with structured tree based protocol}
\label{sec:tree}

In this section, we compare the performance of our protocol with a structured approach for one-shot duplicate insensitive data aggregation that involves maintaining network structures such as spanning trees. For our comparison, we use a prototype tree-based protocol that we describe briefly. The idea is very similar to other tree-based aggregation protocols developed for static sensor networks \cite{tag, ctp}, but the key difference is that the tree is periodically refreshed to handle mobility as described below.

\begin{figure*}[htbp]
  \begin{center}
    \mbox{
      \subfigure[Total messages as a function of node speed] {\scalebox{0.4}{\includegraphics[width=\textwidth]{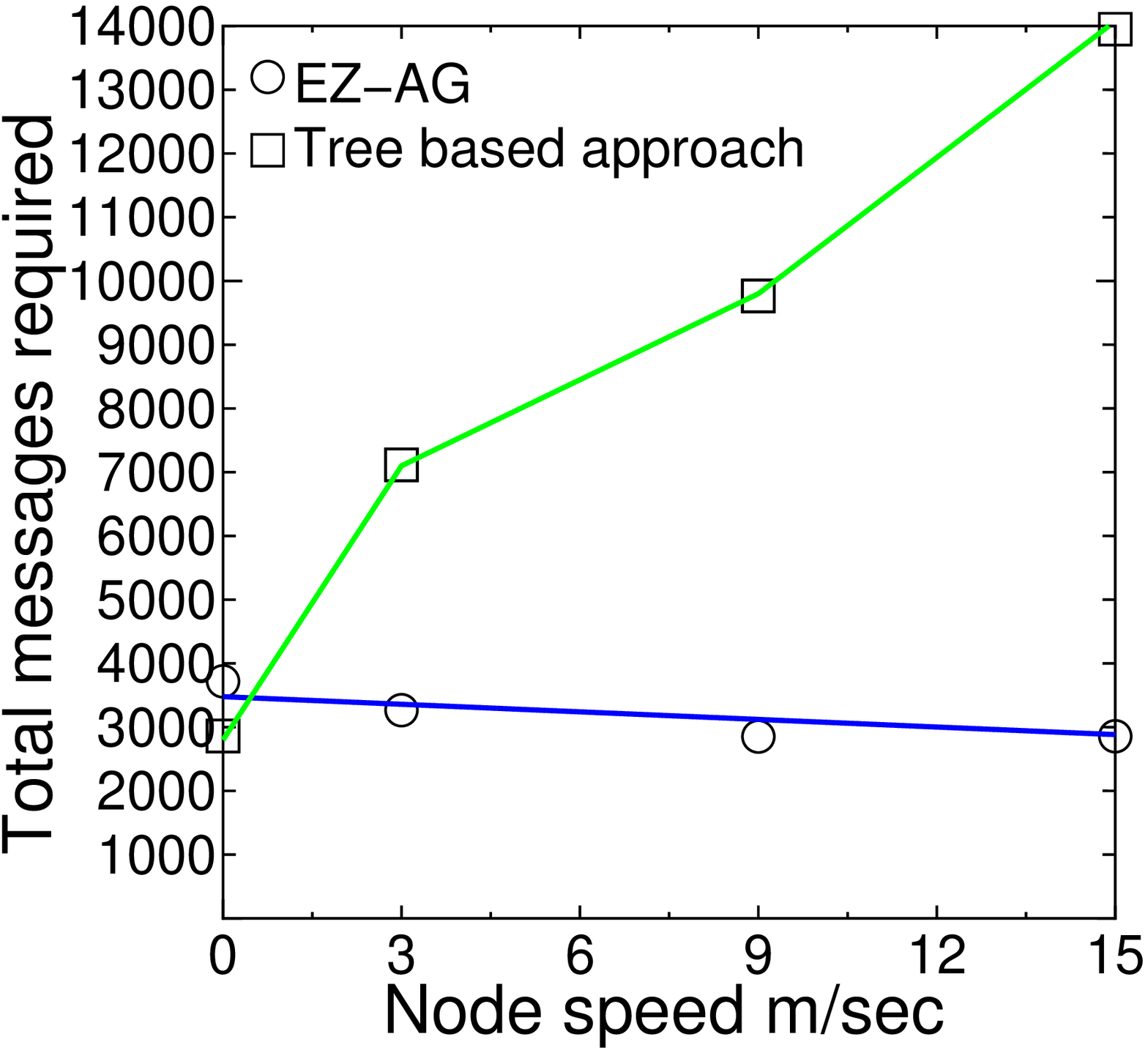}} \label{fig:6a}} \quad
      \subfigure[Aggregation time as a function of node speed] {\scalebox{0.39}{\includegraphics[width=\textwidth]{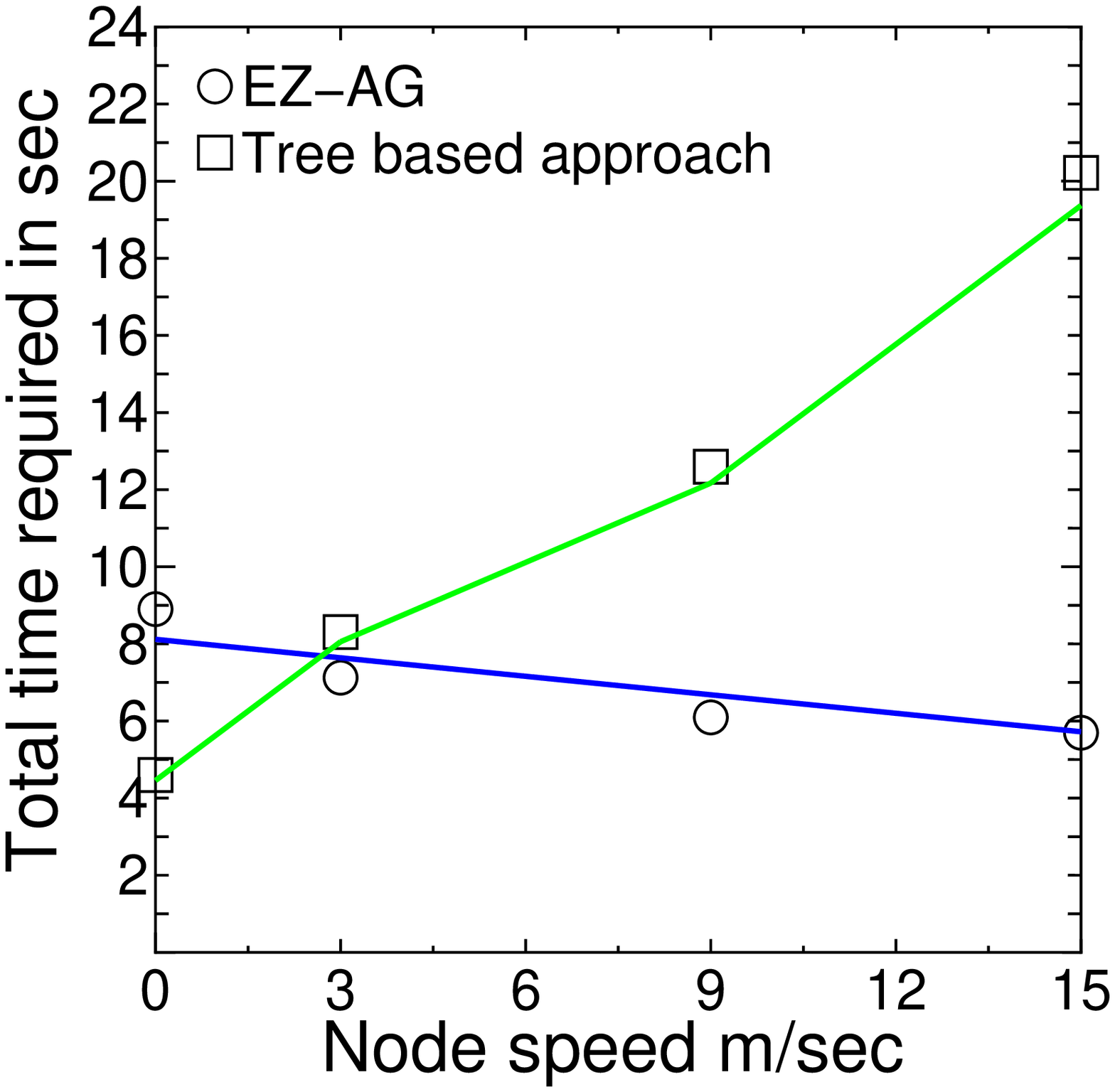}} \label{fig:6b}}         
      } 
    \vspace{-1mm}  
    \caption{{Comparison of time and messages for EZ-AG and tree-based protocol at different node speeds}}
       \label{fig:6}
  \end{center}
\vspace*{-4mm}
\end{figure*}

\begin{figure}[t]
  \begin{center}
    \includegraphics[width=.5\textwidth]{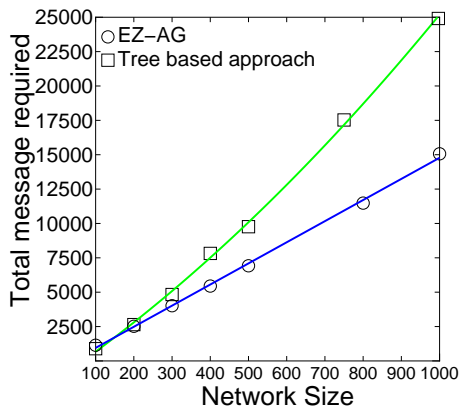}
    \caption{Total messages as a function of network size for EZ-AG and tree based protocol (node speed = 9 m/s) }
    \label{fig:7}
  \end{center}
  \vspace{-3mm}
\end{figure}

The initiating node maintains a tree structure rooted at itself by flooding a request message in the network. Each node maintains a {\em parent} variable. When a node hears a flood message for the first time, it marks the sending node as its {\em parent}. It then schedules a data transmission for its parent at a random time chosen within the next $25$ ms. The message is successively forwarded through the tree structure until it reaches the root. During this process, a node also opportunistically aggregates multiple messages in its transmission queue before forwarding data to its parent. A message could be lost because a node's parent has moved away or due to collisions. To handle message losses, a node repeats its data transmission to its parent until an acknowledgement is received from its parent. While this basic protocol is sufficient for a static network, the network structure is constantly evolving in a mobile network. Hence, the initiating node periodically refreshes the tree by broadcasting a new request every $2$ seconds (with a monotonically increasing sequence number to allow nodes to reset their parents). The refreshing of the tree is stopped when data from all nodes has been received at the initiating node.

In Fig.~\ref{fig:6a}, we compare the total messages required for the tree-based protocol and the random walk based protocol at different node speeds. As seen in this figure, for static networks the tree based protocol is more efficient. However as the mobility increases, the random walk based protocol starts increasing in efficiency. In Fig.~\ref{fig:6b} we compare the total aggregation time which also exhibits a similar trend.

\begin{figure}[t]
  \begin{center}
    \includegraphics[width=.6\textwidth]{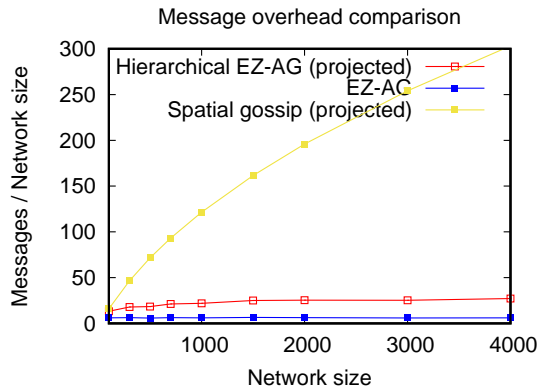}
    \caption{Projected number of messages per node for hierarchical EZ-AG and multi-resolution spatial gossip }
    \label{fig:hierarchy_compare}
  \end{center}
  \vspace{-3mm}
\end{figure}

In Fig.~\ref{fig:7} we compare the total number messages as a function of network size at an average speed of $15$ m/s. Here we observe that the self-repelling random walk based protocol exhibits a linear trend while the tree based protocol exhibits a super-linear trend. This is due to the potentially large number of re-transmissions experienced by the tree-based protocol in a mobile network. This graph also shows that EZ-AG is far more scalable with network size under mobility than structure-based techniques for data aggregation.

\subsection{Comparison with gossip techniques}

In \cite{fisheye_gossip} and \cite{cascade_ipsn07}, a spatial gossip technique is described where each node chooses another node in the network (not just neighbors) at random and gossips its state. When this is repeated $O(log^{1+\epsilon}(N))$ times (where $\epsilon > 1$), all nodes in the network learn about the aggregate state. Note that this scheme requires $O(N.polylog(N))$ messages. EZ-AG requires only $O(N)$ messages. 

In \cite{cascade_ipsn07}, an extension to the spatial gossip technique is described which provides a multi-resolution synopsis of the network state at each node. The technique described in \cite{cascade_ipsn07} requires $O(Nlog^{5.4}(N))$ messages. The hierarchical extension of EZ-AG only requires $O(NlogN)$ messages. The difference is actually quite significant at larger network sizes as seen in Fig.~\ref{fig:hierarchy_compare}, where we show the analytically projected difference in messages transmitted per node for hierarchical EZ-AG and multi-resolution spatial gossip \cite{cascade_ipsn07}. 

\section{Conclusions}
\label{sec:conclusion}

In this paper, we have presented a scalable, robust and lightweight protocol for duplicate insensitive data aggregation in MANETs that exploits the simplicity and efficiency of self-repelling random walks. We showed that by complementing self-repelling random walks with a single step push phase, our protocol can achieve data aggregation in $O(N)$ time and messages. In terms of message overhead, our protocol outperforms existing structure free gossip protocols by a factor of $log(N)$. We quantified the performance of our protocol using ns-3 simulations under different network sizes and mobility models. We also showed that our protocol outperforms structure based protocols in mobile networks and the improvement gets increasingly significant as average node speed increases. 

Importantly, EZ-AG is lightweight in terms of resource requirements and makes rather minimal assumptions of the underlying network.  In particular, it does not assume knowledge of node addresses or locations, require a neighborhood discovery service or network topology information, or depend upon any particular routing or transport protocols such as TCP/IP. 

We also described a hierarchical extension to EZ-AG that provides multi-resolution aggregates of the network state to each node. It outperforms existing technique by a factor of $O(log^{4.4}N)$ in terms of number of messages.

\balance
\bibliographystyle{unsrt}
\bibliography{vinod}

\begin{thebibliography}{10}

\bibitem{cascade_ipsn07}
R.~Sarkar, X.~Zhu, and J.~Gao.
\newblock Hierarchical spatial gossip for multiresolution representations in
  sensor networks.
\newblock In {\em {IPSN}}, pages 311--319, 2007.

\bibitem{fisheye_gossip}
D.~Kempe, J.~M. Kleinberg, and A.~J. Demers.
\newblock {Spatial gossip and resource location protocols}.
\newblock In {\em {ACM Symposium on Theory of Computing}}, pages 163--172,
  2001.

\bibitem{snapshots}
V.~Kulathumani and A.~Arora.
\newblock Distance sensitive snapshots in wireless sensor networks.
\newblock In {\em {Principles of Distributed Systems (OPODIS)}}, volume 4878,
  pages 143--158, 2007.

\bibitem{synopsis_diffusion}
S.~Nath, P.~Gibbons, S.~Seshan, and Z.~Anderson.
\newblock Synopsis diffusion for robust aggregation in sensor networks.
\newblock In {\em Proceedings of the 2Nd International Conference on Embedded
  Networked Sensor Systems}, SenSys '04, pages 250--262, 2004.

\bibitem{sensor-databases}
J.~Considine, F.~Li, G.~Kollios, and J.~Byers.
\newblock Approximate aggregation techniques for sensor databases.
\newblock In {\em Proceedings of the 20th International Conference on Data
  Engineering}, ICDE '04, 2004.

\bibitem{naik_sprinkler}
V.~Naik, A.~Arora, P.~Sinha, and H.~Zhang.
\newblock {Sprinkler: A Reliable and Energy Efficient Data Dissemination
  Service for Extreme Scale Wireless Networks of Embedded Devices}.
\newblock {\em IEEE Transactions on Mobile Computing}, 6(7):777--789, 2007.

\bibitem{tag}
S.~Madden, J.~M. Franklin, J.~Hellerstein, and W.~Hong.
\newblock {TAG: A Tiny AGgregation Service for Ad-hoc Sensor Networks}.
\newblock {\em SIGOPS Oper. Syst. Rev.}, 36(SI):131--146, December 2002.

\bibitem{ctp}
O.~Gnawali, R.~Fonseca, K.~Jamieson, D.~Moss, and P.~Levis.
\newblock Collection tree protocol.
\newblock In {\em Proceedings of the 7th ACM Conference on Embedded Networked
  Sensor Systems}, SenSys '09, pages 1--14, 2009.

\bibitem{directdiff}
C.~Intanogonwiwat, R.~Govindan, D.~Estrin, J.~Heidemann, and F.~Silva.
\newblock Directed diffusion for wireless sensor networking.
\newblock {\em IEEE Transactions on Networking}, 11(1):2--16, 2003.

\bibitem{rtswjournal}
V.~Kulathumani, A.~Arora, M.~Sridharan, K.~Parker, and B.~Lemon.
\newblock On the repair time scaling wall for manets.
\newblock {\em IEEE Communications Letters}, PP(99):1--4, 2016.

\bibitem{mmgossip}
Y.~Chen, S.~Shakkottai, and J.~Andrews.
\newblock On the role of mobility on multi-message gossip.
\newblock {\em {IEEE Transactions on Information Theory}}, 56(12):3953--3970,
  2013.

\bibitem{trickle}
P.~Levis, N.~Patel, S.~Shenker, and D.~Culler.
\newblock Trickle: A self-regulating algorithm for code propagation and
  maintenance in wireless sensor networks.
\newblock In {\em USENIX/ACM Symposium on Networked Systems Design and
  Implementation (NSDI)}, pages 15--28, 2004.

\bibitem{self-repelling}
C.~Byrnes and A.~J. Guttman.
\newblock On self-repelling random walks.
\newblock {\em Journal of Physics A: Mathemaical and General},
  17(17):3335--3342, 1984.

\bibitem{rw12}
H.~Feund and P.~Grassberger.
\newblock How a random walk covers a finite lattice.
\newblock {\em {Physica}}, A(192):465--470, 1993.

\bibitem{gossip-survey}
R.~Friedman, D.~Gavidia, L.~Rodrigues, A.~Viana, and S.~Voulgaris.
\newblock Gossiping on manets: The beauty and the beast.
\newblock {\em SIGOPS Oper. Syst. Rev.}, 41(5):67--74, October 2007.

\bibitem{randgossip}
S.~Boyd, A.~Ghosh, B.~Prabhakar, and D.~Shah.
\newblock Randomized gossip algorithms.
\newblock {\em {IEEE Trans. Info. Theory}}, 52(6):2508–--2530, 2006.

\bibitem{rabbatgossip}
M.~G. Rabbat.
\newblock On spatial gossip algorithms for average consensus.
\newblock In {\em 2007 IEEE/SP 14th Workshop on Statistical Signal Processing},
  pages 705--709, 2007.

\bibitem{rw1}
L.~Lovascz.
\newblock {Random walks on graphs: A survey}.
\newblock Combinatorics, Paul Erdos 80, 1993.

\bibitem{rw7}
G.~Ercal and C.~Avin.
\newblock On the cover time of random geometric graphs.
\newblock {\em {Automata, Languages, and Programming}}, 3580(1):677--689, 2005.

\bibitem{rw8}
C.~Avin and B.~Krishnamachari.
\newblock The power of choice in random walks: An empirical study.
\newblock In {\em Proceedings of the 9th ACM International Symposium on
  Modeling Analysis and Simulation of Wireless and Mobile Systems}, MSWiM '06,
  pages 219--228, 2006.

\bibitem{census-arxiv}
{Hidden}.
\newblock Census: {A} protocol for visiting all nodes in manets using biased
  random walks, 2015.

\bibitem{clementi1}
A.~Clementi, A.~Monti, F.~Pasquale, and R.~Silvestri.
\newblock Information spreading in stationary markovian evolving graphs.
\newblock {\em IEEE Transactions on Parallel and Distributed Systems},
  22(9):1425--1432, 2011.

\bibitem{rw10}
J.~Y.~Le Boudec and M.~Vojnovic.
\newblock Perfect simulation and stationarity of a class of mobility models.
\newblock In {\em IEEE 24th Annual Joint Conference of the IEEE Computer and
  Communications Societies (INFOCOM)}, volume~4, pages 2743--2754 vol. 4, 2005.

\bibitem{rw11}
P.~Nain, D.~Towsley, B.~Liu, and Z.~Liu.
\newblock Properties of random direction models.
\newblock In {\em IEEE 24th Annual Joint Conference of the IEEE Computer and
  Communications Societies (INFOCOM)}, volume~3, pages 1897--1907 vol. 3, 2005.

\bibitem{mobilitymodels-survey}
T.~Camp, J.~Boleng, and V.~Davies.
\newblock A survey of mobility models for ad hoc network research.
\newblock {\em Wireless Communication and Mobile Computing (WCMC): Special
  issue on mobile ad-hoc networking}, 2:483--502, 2002.

\bibitem{srrw-short-arxiv}
{V. Kulathumani and M. Nakagawa and A. Arora}.
\newblock {Coverage characteristics of self-repelling random walks in mobile
  ad-hoc networks}.
\newblock {http://www.community.wvu.edu/~vkkulathumani/srrw.pdf}.

\bibitem{coupon1}
D.J. Newman and L.~Shepp.
\newblock The double dixie cup problem.
\newblock {\em American Math Monthly}, 67(1):58--61, 1960.

\bibitem{collabcounting}
W.~Zeng, A.~Arora, and K.~Srinivasan.
\newblock Low power counting via collaborative wireless communications.
\newblock In {\em Proceedings of the 12th International Conference on
  Information Processing in Sensor Networks}, IPSN '13, pages 43--54, 2013.

\end{thebibliography}


\end{document}